%% file: twoBuyerMultiunit.tex
\documentclass[11pt,oneside]{article}
\usepackage{graphicx}
\usepackage{fullpage}
\usepackage{amsmath}
\usepackage{amssymb}
\usepackage{amsthm}
\usepackage{tikz}
\usetikzlibrary{shapes}
\usetikzlibrary{patterns}
\usetikzlibrary{arrows,positioning,decorations.pathmorphing,trees}
\usetikzlibrary{arrows.meta}
\usepackage{pgfplots}
\usepackage{float}
\usepackage{hyperref}

\usepackage{framed}
\usepackage[utf8]{inputenc}
\usepackage{graphicx,sidecap,tikz}
\usepackage{xcolor}
\definecolor{shadecolor}{gray}{0.9}

\newtheorem{theorem}{Theorem}[section]

\newtheorem{corollary}[theorem]{Corollary}

\newtheorem{lemma}[theorem]{Lemma}
\newtheorem{claim}[theorem]{Claim}

\newtheorem*{remark*}{Remark}
\newtheorem*{notation*}{Notation}
\newtheorem*{observation*}{Observation}
\newtheorem*{theorem*}{Theorem}
\newtheorem*{definition*}{Definition}
\newtheorem*{axiom*}{Axiom}
\newtheorem*{claim*}{Claim}
\newtheorem*{lemma*}{Lemma}

\title{The Price of Anarchy of\\[0.1cm] Two-Buyer Sequential Multiunit Auctions}
\author{Mete \c{S}eref Ahunbay \and Adrian Vetta}

\newcommand{\x}{{\bf x}}
\newcommand{\h}{\mathbb{H}}
\newcommand{\buyset}{\{1,2\}}
\newcommand{\e}{{\bf e}}
\newcommand{\sw}{\textsc{sw}}
\newcommand{\opt}{\textsc{opt}}
\newcommand{\0}{{\bf 0}}

\begin{document}

\maketitle

\begin{abstract}
We study the efficiency of sequential multiunit auctions with two-buyers and complete information. 
For general valuation functions, we show that the {\em price of anarchy} is exactly $1/T$ for auctions with $T$ items for sale. 
For concave valuation functions, we show that the 
price of anarchy is bounded below by $1-1/e\simeq 0.632$. This bound is asymptotically tight as the number of items 
sold tends to infinity.
\end{abstract}

\input{introduction.tex}

\input{model.tex}

\input{linearProgramming.tex}

\input{concave.tex}

\input{general.tex}

\end{document}

%% file: introduction.tex
\section{Introduction}
In a {\em sequential multiunit auction}, $T$ identical copies of an item are sold one at a time.
We evaluate the price of anarchy in two-buyer sequential multiunit auctions with complete information, under the standard model 
introduced by Gale and Stegeman~\cite{GS01}.
Our main result is that, for concave valuation functions, the price of anarchy is at least $1-1/e \simeq 0.632$, and this bound
is asymptotically tight as the number of items $T$ tends to infinity.
We also show that, for general valuation functions, the price of anarchy is exactly $1/T$ for sequential multiunit
auctions with $T$ items for sale.
To obtain these results we show how to lower bound the price of anarchy via a linear programming formulation.
Key to our analyses is a detailed examination of the properties of equilibria. These properties lead to a collection of valid constraints 
whose incorporation into the linear program produces the optimal lower bounds. The optimality of these bounds is certified by 
providing examples of two-buyer sequential auctions with matching upper bounds on the price of anarchy.

\subsection{Related Work}
There is an extensive literature studying the price of anarchy of sequential multiunit auctions. For the case of identical items,
our price of anarchy bound of $1-1/e$ for two-buyer auctions with concave valuations has previously been claimed by Bae et al. \cite{BBB08,BBB09}.
However, those papers contain flaws and the proofs do not hold; see~\cite{ALV20} for details.
Recently, Ahunbay et al.~\cite{ALV20} were able to prove that the price of anarchy is $1-1/e$ under the restriction 
that a buyer may not bid higher than its incremental value for winning the next item. But, even with concave valuations,
equilibrium bids can be higher than incremental values -- thus the results of~\cite{ALV20} do not apply to the traditional equilibrium
concept studied in this paper. 
The price of anarchy of sequential auctions with non-identical items has also been studied in depth; see, for example~\cite{FLS13,PST12,ST13,ST12}. 

To evaluate the price of anarchy we apply primal-dual methods. Nguyen~\cite{Nguyen19} provided through a primal-dual formulation price of anarchy bounds 
for sequential second-price sponsored search auctions, and for sequential first-price auctions with unit-demand valuations. 
Primal-dual methods have also been applied to inspect the efficiency of other classes of games. 
For example, Nadav and Roughgarden~\cite{NR10} by a primal-dual argument characterized the set of outcomes for which 
smoothness~\cite{Rough09} arguments apply for price of anarchy bounds, and proposed a refinement of it to obtain better price of anarchy 
bounds for coarse correlated equilibria. Bilo~\cite{Bilo13} showed that constant-ratio efficiency 
bounds may be obtained for weighted congestion games even with quadratic and cubic latency functions through a primal-dual formulation. 
Likewise, Kulkarni and Mirrokni~\cite{KM15} provided bounds on the robust price of anarchy for several classes of games through the use of 
LP and Fenchel duality. 

\subsection{Overview}
Section~\ref{sec:model} describes the two-buyer sequential multiunit auction. There some basic properties of equilibria are given
along with an example. 
 In Section~\ref{sec:LP}, we present a linear programming approach for lower bounding the price of anarchy in
these auctions. We then investigate structural properties of equilibria that induce valid inequalities that can be incorporated into
our linear programs. In Section~\ref{sec:concave} we show via LP-duality that this approach produces a bound of $1-1/e$ for concave 
valuation functions; we then show this bound is asymptotically tight.
Finally, in Section~\ref{sec:general} this method is used produce a lower bound of $1/T$ for general valuation functions. An example with a
matching upper bound shows this price of anarchy bound is tight.

%% file: model.tex
\section{The Sequential Auction Model}\label{sec:model}
We study two-buyer sequential auctions under the complete information model of Gale and Stegeman~\cite{GS01}.
Here, we present the model, notation and the concept of efficiency as in Ahunbay et al.~\cite{ALV20}.
There are $T$ identical items which are sold one by one in a sequence of second-price auctions. 
Buyer~$i \in \buyset$ has \emph{value} $V_i(k)$ for winning $k$ items. Given the valuations, we will say buyer $i$ 
has \emph{incremental value} $v_i(k)$
for obtaining a $k$th item: formally, $v_i(k) = V_i(k) - V_i(k-1)$. 
We also make the standard assumption of \emph{free disposal}. Thus, the valuation functions
are non-decreasing; in particular, the incremental values are non-negative, i.e. $v_i(k) \geq 0$, for any buyer $i$ and 
any $k \in [T]:=\{1,2,\dots, T\}$.
Furthermore, we say the valuation function is {\em concave} if the incremental values are non-increasing; that is
$v_i(k)\ge v_i(k+1)$ for any $k\in [T]$.

\subsection{Forward Utilities and Equilibria}
To find an equilibrium in the sequential auction we make a Markov perfection assumption: in each round of the auction, 
buyer $i$ makes a bid conditioned on the 
number of items previously won by each buyer. The set of histories is then given by 
$\h = \{ \x \in \mathbb{Z}^2: \x \geq 0, x_1 + x_2 \leq T\}$. For $\x \in \h$, if $x_1 + x_2 = T$, then $\x$ is 
called a \emph{terminal node}; otherwise, $\x$ is called a \emph{decision node}.

We can find an {\em equilibrium} by computing the \emph{forward utility} of each buyer at each node 
$\x \in \h$. The forward utility at $\x$ is the profit a buyer will earn from period $x_1 + x_2$ 
onwards, provided that each buyer $i$ has won $x_i$ items. 
This can be calculated by backwards induction on $x_1 + x_2$. 
If $\x$ is a terminal node then the auction has 
ended. Hence the forward utility of each buyer is zero at such a terminal node, i.e. $u_i(\x) = 0$ for any $i \in \buyset$. 
It remains to evaluate the forward utility of each buyer at a decision node $\x$. Decision node $\x$ has two \emph{direct successors}: 
the node $(x_1+1,x_2)$ is called the {\em left child} of $\x$ and corresponds 
to buyer~$1$ winning the item at $\x$; likewise the node $(x_1,x_2+1)$ is the {\em right child} of $\x$ and corresponding to buyer~$2$ winning
at $\x$. Then, in a second-price auction at decision node $\x$, the unique bidding strategies that survive the iterative elimination 
of weakly dominated strategies are:
\begin{align}\label{eqn:bids}
b_1(\x) &\ =\  v_1(x_1+1) + u_1(x_1+1,x_2) - u_1(x_1,x_2+1) \nonumber \\
b_2(\x) &\ =\ v_2(x_2+1) + u_2(x_1,x_2+1) - u_2(x_1+1,x_2)
\end{align}
Let $p(\x)$ denote the price paid by the winning buyer at decision node $\x$. As this is a second-price auction, this 
price is simply the minimum of the two bids: 
\begin{equation}\label{eqn:price}
p(\x) \ =\  \min_{i \in \buyset}  b_i(\x)
\end{equation}
Now, if $b_1(\x) > b_2(\x)$ then buyer~$1$ wins and the utilities of the buyers are given by:
\begin{align}\label{eqn:1winUtilities}
u_1(\x) & = v_1(x_1+1)  - b_2(\x) + u_1(x_1+1,x_2) \nonumber \\
& = v_1(x_1+1) + u_1(x_1+1,x_2) + u_2(x_1+1,x_2) - u_2(x_1,x_2+1) - v_2(x_2+1) \nonumber \\
u_2(\x) & = u_2(x_1+1,x_2) 
\end{align}
Conversely, if $b_1(\x) < b_2(\x)$ then buyer~$2$ wins, and the utilities are defined symmetrically as:
\begin{align}\label{eqn:2winUtilities}
u_1(\x) & = u_1(x_1,x_2+1) \nonumber \\
u_2(\x) & = v_2(x_2+1)  - b_1(\x) + u_2(x_1,x_2+1) \nonumber \\
& = v_2(x_2+1) + u_1(x_1,x_2+1) + u_2(x_1,x_2+1) - u_1(x_1+1,x_2) - v_1(x_1+1)
\end{align}
Finally, if $b_1(\x) = b_2(\x)$ then  for any buyer $i$:
\begin{equation}\label{eqn:tieUtilities}
u_i(\x+\e_{-i}) = v_i(x_i+1) - b_i(\x) + u_i(\x+\e_i)
\end{equation}
Thus in the case of a tie, the utilities are invariant to which way the tie is broken. In particular, the forward utilities and 
bids of the buyers at each node are uniquely determined. Observe that this means that the Markov perfection assumption 
does not result in any loss of generality. 
We remark that a two-buyer sequential multiunit auction may be represented by a labelled directed tree rooted at decision 
node $(0,0)$. Such a representation is later given in Example~1.
This notation allows for a simple description of the forward utilities at each decision node. 
Denote by $U(\x) = \sum_{i \in \buyset} u_i(\x)$ the sum of the forward utilities of the two buyers at node $\x=(x_1,x_2)$.
\begin{claim}\label{cl:i-wins}
Let $\x$ be a decision node. Then $b_i(\x) \geq b_{-i}(\x)$ if and only if:
$$v_i(x_i+1) + U(\x+\e_i)  \ \geq\ v_{-i}(x_{-i}+1) + U(\x+\e_{-i})$$
\end{claim}
\begin{proof}
By (\ref{eqn:bids}), we have $b_i(\x) \geq b_{-i}(\x)$ at decision node $\x$ if and only if:
$$v_i(x_i+1) + u_i(\x+\e_i) - u_i(\x+\e_{-i}) \ \geq\ v_{-i}(x_{-i}+1) + u_{-i}(\x+\e_{-i}) - u_{-i}(\x+\e_i)$$
Rearranging, this is equivalent to:
$$v_i(x_i+1) + u_i(\x+\e_i) + u_{-i}(\x+\e_i) \ \geq\ v_{-i}(x_{-i}+1) + u_i(\x+\e_{-i}) + u_{-i}(\x+\e_{-i})$$
The claim then follows by definition of $U(\x+\e_i)$ and $U(\x+\e_{-i})$. 
\end{proof}
Together with (\ref{eqn:1winUtilities}) and (\ref{eqn:2winUtilities}), Claim~\ref{cl:i-wins} implies the following.
\begin{claim}[\cite{GS01}, Equation 7]\label{cl:maxForm}
The forward utility of buyer~$i$ at decision node $\x$ is exactly:
$$u_i(\x) = \max_{j \in \buyset} [v_j(x_j+1) + U(\x+\e_j)] - v_{-i}(x_{-i}+1) - u_{-i}(\x + \e_{-i}) \eqno\qed$$
\end{claim}
To conclude this section, we illustrate the auction and associated concepts by an example.
\begin{shaded}
\noindent \textbf{Example 1.}
Consider a sequential multiunit auction of $T = 2$ items. Let buyer~$1$ and buyer~$2$ have concave valuation functions with
incremental values $\left(v_1(1),v_1(2)\right) = (10,10)$ and $\left(v_2(1),v_2(2)\right) = (5,0)$. This auction has three decision 
nodes $\{(0,0), (1,0), (0,1)\}$ and 
three terminal nodes  $\{(2,0), (1,1), (0,2)\}$. At decision nodes $(1,0)$ and $(0,1)$, buyers face a single-item 
second price auction and 
bid their incremental valuations:
\begin{align*}
b_1(1,0) & = v_1(2) = 10& b_1(0,1) & = v_1(1) = 10 \\
b_2(1,0) & = v_2(1) = 5& b_2(0,1) & = v_2(2) = 0
\end{align*}
Therefore, at decision node $(1,0)$ buyer $1$ wins the item, and the auction proceeds to terminal node $(2,0)$. 
Buyer $1$ also wins the item at decision node $(0,1)$, with the auction then
proceeding to terminal node $(1,1)$. At decision node $(1,0)$ buyer $1$ pays $b_2(1,0) = 5$ 
for the item but at decision node $(0,1)$ pays $b_2(0,1) = 0$. Hence the forward utilities of the buyers at 
the decision nodes $(1,0)$ and $(0,1)$ are:
\begin{align*}
u_1(1,0) &= v_1(2) - b_2(1,0) + u_1(2,0) = 10 - 5 + 0 = 5 & u_2(1,0) = u_2(2,0) = 0\\
u_1(0,1) &= v_1(1) - b_2(0,1) + u_1(1,1)  = 10 - 0 + 0 = 10 & u_2(0,1) = u_2(2,0) = 0
\end{align*}
Now consider decision node $(0,0)$. Here the buyers bid strategically, knowing 
that whether they win or lose in the first period changes their profits throughout the auction:
\begin{align*}
b_1(0,0) &= v_1(1) + u_1(1,0) - u_1(0,1) = 10 + 5 - 10 = 5 \\
b_2(0,0) &=  v_2(1) + u_2(0,1) - u_2(1,0) = 5 + 0 - 0 = 5
\end{align*}
Thus the bids of the buyers tie at $(0,0)$. If the tie is broken in favour of 
buyer $1$ at price $5$, its forward utility is:
\begin{align*}
u_1(0,0) & = v_1(1) - b_2(0,0) + u_1(1,0) = 10 - 5 + 5 = 10
\end{align*}
If the tie is broken in favour of buyer $2$ then $u_1(0,0) = u_1(0,1) = 10$. Therefore, buyer $1$'s utility indeed is the same 
regardless of which way the tie is broken. Similarly, $u_2(0,0) = 0$ irrespective of which way the tie is broken. 
The resulting auction tree is:

\begin{center}
 \begin{tikzpicture}[scale=.7]
  \node [ellipse,draw, fill=red!20, align=center, scale=0.6] (v1) at (4,4) { {\bf (0,0)} \\ {\bf 10} : {\bf 0} };
\node [ellipse,draw, fill=red!20, align=center, scale=0.6](v2) at (2,2) { {\bf(1,0)} \\ {\bf 5} : {\bf 0} };
\node [ellipse,draw, fill=red!20, align=center, scale=0.6](v3) at (6,2) { {\bf(0,1)} \\ {\bf 10} : {\bf 0} };
\node[ellipse,draw, fill=yellow!20, align=center, scale=0.6](v4) at (0,0) { {\bf (2,0)} \\ {\bf 0} : {\bf 0}};
\node[ellipse,draw, fill=yellow!20, align=center, scale=0.6](v5) at (4,0) { {\bf (1,1)} \\  {\bf 0} : {\bf 0} };
\node[ellipse,draw, fill=yellow!20, align=center, scale=0.6](v6) at (8,0) { {\bf (0,2)} \\ {\bf 0} : {\bf 0} };
\draw [->] (v1) -- (v2) node[very near start,left, scale = .66]{$5\ $};
\draw [->, very thick] (v1) -- (v3) node[very near start,right, scale = .66]{$\ \ 5$};
\draw [->] (v2) -- (v4) node[very near start,left, scale = .66]{$10\ $} ;
\draw [->, dotted] (v2) -- (v5) node[very near start,right, scale = .66]{$\ 5$};
\draw [->, very thick] (v3) -- (v5) node[very near start,left, scale = .66]{$10\ $};
\draw [->, dotted] (v3) -- (v6) node[very near start,right, scale = .66]{$\ \ 0$};
\end{tikzpicture}
\end{center}

\noindent At each node $(x_1,x_2)$ the forward utilities of the buyers are shown. The label of the arrow from a decision node to its left child 
shows buyer $1$'s bid at that decision node; similarly, the label of the arrow to its right child shows buyer $2$'s bid. 
An arrow is solid if its label is a maximal bid at the node it originates from, and dotted otherwise. The bold path shows the realization of the 
equilibrium path if the tie at node $(0,0)$ is broken in favour of buyer~$2$.
\end{shaded}

\subsection{Social Welfare and Efficiency}\label{subsec:swAndEfficiency}

The purpose of this paper is to evaluate the price of anarchy in the sequential auction. This requires us to formally define
the social welfare of an allocation. To wit, let $\x$ be a decision node and denote by $t(\x) = T-x_1-x_2$ the number of items for sale 
starting from $\x$. Then the \emph{social welfare} from decision node $\x$ of the allocation where buyer~$1$ wins exactly $k$ more items is: 
\begin{align}\label{def-sw}
\sw(k|\x) 
&\ =\ V_1(x_1+k)  - V_1(x_1) + V_2(T-x_1-k) - V_2(x_2) \nonumber \\
&\ =\ \sum_{j = x_1+1}^{x_1+k} v_1(j) + \sum_{j = x_2 + 1}^{T-x_1-k} v_2(j)
\end{align}
The \emph{optimal social welfare} from decision node $\x$ is then given by:
\begin{equation}\label{def-opt-sw}
\opt(\x) \ =\ \max_{k \in [t(\x)] \cup \{0\}} \,\sw(k|\x)
\end{equation}
Our formal treatment of \emph{efficiency} will relate the optimal social welfare at a decision node $\x$ to the social 
welfare of some terminal node $\x + (k,t(\x)-k)$, given that there exists some equilibrium path connecting the two nodes.
To do this, we first present the formal definition of an \emph{(equilibrium) path} from \cite{ALV20}. 
A \emph{path} from decision node $\x$ is a $(t(\x)+1)$-tuple $P = (\x^{t(\x)},\x^{t(\x)-1},...,\x^1,\x^0)$ such that: (i) the 
path starts from $\x$, i.e. $\x^{t(\x)} = \x$, and (ii) each successive node follows from some buyer~$j$ acquiring an 
item, i.e. for each $k \in [t(\x)]$, $\x^{k-1} = \x^{k}+\e_j$ for some $j \in \buyset$. A path $P$ is called 
an \emph{equilibrium path} if each successive node follows from some buyer $j$ acquiring an item by outbidding 
the other player, i.e. for any $k \in [t(\x)]$, if $\x^{k-1} = \x^{k}+\e_j$, then $b_j(\x^k) \geq b_{-j}(\x^k)$. Finally, for 
any path $P = (\x^{t(\x)},\x^{t(\x)-1},...,\x^1,\x^0)$ and any $s \in [t(\x)]$, we denote by $P^s = (\x^s,\x^{s-1},...,\x^1,\x^0)$ 
the final segment of $P$ of $s+1$ nodes.

We now present our notion of efficiency. The \emph{efficiency along path $P$}, denoted $\Gamma(P)$, satisfies:
\begin{equation}
\Gamma(P) = \begin{cases}
\,\frac{\sw(x_1^0-x_1^{t(\x)} | \x)}{\opt(\x)} &\qquad \opt(\x) \neq 0 \\
\qquad 1 &\qquad \opt(\x) = 0
\end{cases}
\end{equation}
For example, consider the equilibrium path shown in the figure of Example 1. This path from the root of the tree allocates one item to each buyer
so $\sw(x_1^0-x_1^{t(\0)} | \0)=5+10=15$. On the hand social welfare is maximized by allocated both items to buyer~$1$, so $\opt(\0)=10+10=20$. 
Consequently, the efficiency along this path is $\frac{15}{20}=\frac34$.
The \emph{price of anarchy} (over some class) is then the infimum of the set of possible efficiency values along 
equilibrium paths of auctions in class. Here, of course, the class of auctions we consider is two-buyer sequential multiunit auctions.

We remark that because the valuation functions $V_i(\cdot)$ are non-decreasing, the price of anarchy is meaningful; in particular,
it is always lies between $0$ and $1$. To see this, as the incremental valuations are non-negative, (\ref{def-sw}) and (\ref{def-opt-sw}) 
imply that $\opt(\x) \geq \sw(k|\x) \geq 0$ at any decision node $\x$, for any $0 \leq k \leq t(\x)$. 
Thus $\Gamma(P) \in [0,1]$ for any path $P$ starting from $\x$.

%% file: linearProgramming.tex
\section{A Linear Programming Formulation}\label{sec:LP}

In this section, we provide a linear programming approach for bounding the price of anarchy in
a two-buyer sequential auction.
We begin in Section~\ref{sec:structure} by presenting a set of structural results concerning equilibria in
the sequential auction. These structural properties will induce a class of linear programs that can 
be used to lower bound the price of anarchy. Then, in Section~\ref{sec:valid}, we motivate and generate an 
additional class of valid inequalities that must hold along equilibrium paths.
In Sections~\ref{sec:concave} and~\ref{sec:general}, we will prove that the incorporation of these valid inequalities into our linear programs
suffice to provide tight price of anarchy bounds for concave valuation functions and general valuation functions, respectively.

\subsection{Structural Results}\label{sec:structure}

Let us first note two results from Gale and Stegeman~\cite{GS01}. The first is the intuitive result that buyer~$i$ does not derive any 
benefit from letting buyer~$-i$ win an item at no cost:

\begin{lemma}[\cite{GS01}, Lemma 1]\label{lem:noFreeWin}
Let $\x$ be a decision node. Then $u_i(\x) \geq u_i(\x+\e_{-i})$ for any buyer $i$. Moreover, this inequality is 
strict if and only if $b_{i}(\x) > b_{-i}(\x)$, that is, buyer $i$ wins with a strictly greater bid at decision node $\x$. \qed
\end{lemma}

The second result is the {\em declining price anomaly}: prices are non-increasing along any equilibrium path.
\begin{lemma}[\cite{GS01}, Lemma 2]\label{lem:decreasingPrices}
Let $\x$ be a decision node such that $t(\x) > 1$. If buyer $i$ wins at $\x$, then $p(\x) \geq p(\x+\e_i)$. 
Moreover, $p(\x) = p(\x+\e_i)$ if and only if buyer $i$ also wins at decision node $\x+\e_{-i}$. \qed
\end{lemma}

Because the forward utilities are zero at each terminal node, an immediate consequence
of Lemma~\ref{lem:noFreeWin} is that the forward utility of any buyer at any decision node is non-negative.
\begin{corollary}\label{cor:nonNegUtils}
Let $\x \in \h$ and $i \in \buyset$, then $u_i(\x) \geq 0$. \qed
\end{corollary}

Furthermore, given the assumption of non-decreasing valuation functions, Lemma \ref{lem:decreasingPrices} implies 
that prices are non-negative.
\begin{lemma}\label{lem:nonNegativePrices}
At any decision node $\x$, $p(\x) \geq 0$.
\end{lemma}
\begin{proof}
By induction on $t(\x)$. For the base case, $t(\x) = 1$, we have a single-item second price auction at $\x$. 
Therefore, $p(\x) = \min_{i \in \buyset} v_i(x_i+1) \geq 0$ as the incremental valuations are non-negative. 
Now consider $t(\x) > 1$. Let buyer $i$ win at $\x$. Then, by Lemma~\ref{lem:decreasingPrices}, $p(\x) \geq p(\x+\e_i)$. 
But, by the induction hypothesis, $p(\x+\e_i) \geq 0$. Thus, $p(\x) \geq 0$.
\end{proof}

We may now derive a simple upper bound on the forward utility of any buyer.
Observe that, for any equilibrium path $P$ starting at $\x$, the forward utility of buyer~$i$ at $\x$ is the value
it has for the items it wins on the path $P$ minus the total price it pays. Thus, because the prices are 
non-negative by Lemma~\ref{lem:nonNegativePrices}, the total value buyer $i$ has for the items it wins on the equilibrium path $P$ 
is an upper bound on its forward utility at $\x$.
\begin{lemma}\label{lem:utilityUpperBounds}
Let $\x$ be a decision node and $P = (\x^{t(\x)},...,\x^0)$ an 
equilibrium path starting at $\x$. Then, for any $i \in \buyset$, $u_i(\x) \leq \sum_{j = x_i +1}^{x^0_i} v_i(j)$. \qed
\end{lemma}

Claim~\ref{cl:maxForm} also immediately provides an explicit form for the difference of the buyers' forward utilities.
\begin{lemma}\label{lem:utilDifs}
Let $\x$ be a decision node. Then:
$$u_i(\x) - u_{-i}(\x) = v_i(x_i+1) + u_i(\x+\e_i) - v_{-i}(x_{-i}+1) - u_{-i}(\x+\e_{-i})\eqno\qed$$
\end{lemma}

Finally, we turn our attention to the efficiency of paths. As the valuations are non-decreasing, we can show that the efficiency 
along a path $P = (\x^{t(\x)},\x^{t(\x)-1},...,\x^1,\x^0)$ may be bounded below by the efficiency along a specific 
subpath $P^{s}$ of $P$ (which may be $P$ itself), such that the unique optimum allocation from $\x^s$ has one 
buyer winning all the remaining items. The result generalises arguments made in the proofs of Theorem~2 in~\cite{BBB08}, 
and Lemma~6.2 and Lemma~6.3 in~\cite{ALV20} to possibly non-concave valuations.

\begin{lemma}\label{lem:pathEfficiencyBound}
Let $\x$ be a decision node and $P = (\x^{t(\x)},\x^{t(\x)-1},...,\x^1,\x^0)$ a 
path from $\x$. Then:
\begin{enumerate}
\item[\emph{(a)}] If $\x^{t(\x)-1} = \x + \e_1$ and $\exists k > 0, \sw(k|\x) = \opt(\x)$, then $\Gamma(P) \geq \Gamma(P^{t(\x)-1})$.
\item[\emph{(b)}] If $\x^{t(\x)-1} = \x + \e_2$ and $\exists k < t(\x), \sw(k|\x) = \opt(\x)$, then $\Gamma(P) \geq \Gamma(P^{t(\x)-1})$.
\end{enumerate}
\end{lemma}

\begin{proof}
It suffices to prove (a), as (b) then follows from relabelling the buyers. If $\opt(\x) = 0$ then, because the efficiency is 
between $0$ and $1$ along any path, we have $\Gamma(P) = 1 \geq \Gamma(P^{t(\x)-1})$. So suppose 
that $\opt(x) > 0$.  Note that $t(\x+\e_1) = t(\x)-1$. So, by definition~(\ref{def-sw}), $\sw(k|\x+\e_1) = \sw(k+1|\x) - v_1(x_1+1)$ for 
any $k \in [t(\x+\e_1)]\cup \{0\}$. By assumption, there exists $k > 0$ such that $\sw(k|\x) = \opt(\x)$. Thus
$\opt(\x+\e_1)  = \opt(\x)-v_1(x_1+1)$. Therefore:
\begin{align*}
\Gamma(P) 
&\ =\ \frac{\sw(x^0_1 - x_1|\x)}{\opt(\x)} 
\ =\ \frac{\sw(x^0 - x_1 - 1|\x+\e_1) + v_1(x_1+1)}{\opt(\x+\e_1) + v_1(x_1+1)} 
\end{align*}
If $\opt(\x+\e_1) = 0$ then $\Gamma(P) = 1$ and so $\Gamma(P) \geq \Gamma(P^{t(\x)-1})$, as desired. 
Therefore, we may assume that $\opt(\x+\e_1) > 0$. Then:
\begin{align*}
\Gamma(P) 
&\ =\ \frac{\sw(x^0 - x_1 - 1|\x+\e_1) + v_1(x_1+1)}{\opt(\x+\e_1) + v_1(x_1+1)} \\
&\ \geq\ \frac{\sw(x^0 - (x_1 +1)|\x+\e_1)}{\opt(\x+\e_1)} \\
&\ =\ \Gamma(P^{t(\x)-1})
\end{align*}
Here the inequality holds because $v_1(x_1+1) \geq 0$.
\end{proof}

\begin{corollary}\label{cor:pathEfficiencyBound}
Let $\x$ be a decision node and $P = (\x^{t(\x)},\x^{t(\x)-1},...,\x^1,\x^0)$ a 
path from $\x$. If $\Gamma(P) < \Gamma(P^{t(\x)-1})$ then exactly one of the following holds:
\begin{enumerate}
\item[\emph{(a)}] Buyer~$1$ winning all the items is the unique optimal allocation but buyer~$2$ wins 
an item at $\x$, i.e. $\arg \max_{k \in [t(\x)] \cup \{0\}} \sw(k|\x) = \{t(\x)\}$ and $x^{t(\x)-1} = x^{t(\x)}+\e_2$.
\item[\emph{(b)}] Buyer~$2$ winning all the items is the unique optimal allocation but buyer~$1$ wins 
an item at $\x$, i.e. $\arg \max_{k \in [t(\x)] \cup \{0\}} \sw(k|\x) = \{0\}$ and $x^{t(\x)-1} = x^{t(\x)}+\e_1$.\qed
\end{enumerate}
\end{corollary}

Corollary~\ref{cor:pathEfficiencyBound} implies that, if we are interested in obtaining a lower bound for 
efficiency, it is sufficient to consider auctions in which the unique optimal allocation is $(T,0)$
but buyer~$1$ wins less than $T$ items on an equilibrium path. 
Furthermore, any auction with efficiency less than $1$ must have positive optimal welfare.
Consequently, by multiplying all valuations, forward utilities and bids by a constant,
we may normalize so that $\opt(\0) = 1$. 

The following class of linear programs then provides lower bounds on the efficiency of the 
auction, conditional on buyer~$1$ winning $k$ items on the equilibrium path:

\begin{alignat}{3}\label{LP:unrestricted}
\text{minimize}\  & &\sum_{j = 1}^{k} v_1(j) + \sum_{j = 1}^{T-k} v_2(j) &&& \\
\text{subject to}\ & &\sum_{j = 1}^T v_1(j) &= 1 &&\nonumber \\
&    &\sum_{j = 1}^l v_1(j) + \sum_{j = 1}^{T-l} v_2(j) &\leq 1 &\quad\quad&\forall\, 0 \leq l < T \nonumber \\
&  &v_i(j) &\geq 0 							&\quad\quad&\forall i \in \buyset, j \in [T] \nonumber \\
&&&\text{(+ valid inequalities)}&&\nonumber 
\end{alignat}
%
%
Specifically, we obtain a lower bound for the efficiency of a $T$ item auction by finding the minimum value 
of the linear program for $0 \leq k < T$. Of course, the lower bound produced will depend upon the
choice of additional valid inequalities.
The difficulty is to select inequalities that must be satisfied at equilibrium {\bf and} that are strong enough to
provide an exact efficiency bound. Thus, our task reduces to finding such a set of inequalities.

\subsection{A Set of Valid Inequalities}\label{sec:valid}
The following theorem will allow us obtain a collection of valid inequalities that 
are strong enough to induce tight price of anarchy bounds.

\begin{theorem}\label{thm:pathValid}
Let $\x$ be a decision node and $P$ an equilibrium path from $\x$ with endpoint $(k,T-k)$. Then:
\begin{equation}\label{eqn:validUtil}
\sum_{j = 0}^{t(\x)} u_2(\x+j\e_1) \leq \sum_{i = x_2+1}^{T-k} \Big[(T-x_1-i+1) \cdot v_2(i) - \sum_{j = k+1}^{T-i+1} v_1(j)\Big]
\end{equation}
\end{theorem}

\begin{proof}
Proceed by induction on $t(\x)$. In the base case, $t(\x) = 1$, so either $\x = (k-1,T-k)$ or $\x = (k,T-k-1)$. 
First suppose that $\x = (k-1,T-k)$. Then buyer~$1$ wins at decision node $\x$, and so:
$$u_2(\x) = u_2(\x+\e_1)= u_2(k,T-k)$$
But $(k,T-k)$ is a terminal node and so $u_2(k,T-k) = 0$. This implies that the LHS of (\ref{eqn:validUtil}) is 
$\sum_{j = 0}^{1} u_2(\x+j\e_1)= u_2(\x)+u_2(\x+\e_1)=0$. 
Because $x_2+1=T-k+1$, the RHS of (\ref{eqn:validUtil}) is the empty sum. Thus, the inequality holds with equality. 

Second, suppose  $\x = (k,T-k-1)$. Then buyer~$2$ wins at decision node $\x$. As $(k+1,T-k-1)$ is a terminal node, $u_2(k+1,T-k+1) = 0$. Therefore:
\begin{align*}
\sum_{j = 0}^{t(\x)} u_2(\x+j\e_1) 
& = u_2(k,T-k-1) + u_2(k+1,T-k+1) \\
& = u_2(k,T-k-1) \\
& = v_2(T-k) + u_2(k,T-k) + u_1(k,T-k) - v_1(k+1) - u_1(k+1,T-k-1) \\
& = v_2(T-k) - v_1(k+1) \\
& = \sum_{i = x_2+1}^{T-k} \Big[(T-x_1-i+1) \cdot v_2(i) - \sum_{j = k+1}^{T-i+1} v_1(j)\Big]
\end{align*}
Here the second equality holds because $(k+1,T-k-1)$ is a terminal node and so $u_2(k+1,T-k+1) = 0$.
The third equality holds by definition~(\ref{eqn:2winUtilities}). The fourth equality also holds because the forward utility at 
each terminal node is zero. The last inequality holds as $x_1=k$ and $x_2+1=T-k$. Therefore, (\ref{eqn:validUtil}) again holds with equality.

Now consider $t(\x) > 1$. First suppose that, on an equilibrium path $P$, buyer~$1$ wins at decision node $\x$. Then:
\begin{align*}
\sum_{j = 0}^{t(\x)} u_2(\x+j\e_1) & = u_2(\x) + \sum_{j = 1}^{t(\x)} u_2(\x+j\e_1) \\
& = u_2(\x) + \sum_{j = 0}^{t(\x+\e_1)} u_2\big((\x+\e_1)+j\e_1\big) \\
& \leq u_2(\x) + \sum_{i = x_2+1}^{T-k} \Big[\big(T-(x_1+1)-i+1\big) \cdot v_2(i) - \sum_{j = k+1}^{T-i+1} v_1(j)\Big] \\
& \leq \sum_{i = x_2+1}^{T-k} v_2(i) + \sum_{i = x_2+1}^{T-k} \Big[\big(T-(x_1+1)-i+1\big) \cdot v_2(i) - \sum_{j = k+1}^{T-i+1} v_1(j)\Big] \\
& = \sum_{i = x_2+1}^{T-k} \Big[(T-x_1-i+1) \cdot v_2(i) - \sum_{j = k+1}^{T-i+1} v_1(j)\Big] 
\end{align*}
Here the second equality holds by the fact that $t(\x+\e_1) = t(\x) - 1$. The first inequality follows from the 
induction hypothesis, because $P^{t(\x)-1}$ is an equilibrium path from $\x+\e_1$ to $(k,T-k)$. The second inequality arises from the upper bound on
the forward utility given by Lemma~\ref{lem:utilityUpperBounds}.

Next suppose buyer~$2$ wins at decision node $\x$. Then we have:
\begin{align}\label{eqn:inthm1}
\sum_{j = 0}^{t(\x)} u_2(\x+j\e_1) & = u_2(\x) + \sum_{j = 1}^{t(\x)} u_2(\x+j\e_1)  \\
& = v_2(x_2+1) + u_2(\x+\e_2) + u_1(\x+\e_2) - v_1(x_1 + 1) - u_1(\x+\e_1) + \sum_{j = 1}^{t(\x)} u_2(\x+j\e_1) \nonumber\\
& = v_2(x_2+1) + u_2(\x+\e_2) + u_1(\x+\e_2) - v_1(x_1 + 1) - u_1(\x+\e_1) + \sum_{j = 1}^{t(\x)-1} u_2(\x+j\e_1) \nonumber
\end{align}
Here the second equality holds by definition~(\ref{eqn:2winUtilities}). The third inequality holds as
$u_2(\x + t(\x)\e_1) = 0$ because $\x + t(\x)\e_1$ is a terminal node.

To simplify this we make repeated applications of Lemma~\ref{lem:utilDifs} on $- u_1(\cdot)$. In particular:
\begin{align}\label{eqn:inthm2}
-u_1(\x+\e_1) 
& = -u_2(\x+\e_1) + v_2(x_2+1) + u_2(\x+\e_2+\e_1) - v_1(x_1 + 2) - u_1(\x+2\e_1) \nonumber \\
& = -u_2(\x+\e_1) + v_2(x_2+1) + u_2(\x+\e_2+\e_1) - v_1(x_1 + 2) \nonumber \\
&\qquad -u_2(\x+2\e_1) + v_2(x_2+1) + u_2(\x+\e_2+2\e_1) - v_1(x_1 + 3) - u_1(\x+3\e_1) \nonumber\\
&\ \ \vdots \\
& = - \sum_{j = 1}^{t(\x)-1} u_2(\x+j\e_1) + (t(\x)-1) \cdot v_2(x_2+1) + \sum_{j = 1}^{t(\x)-1} u_2(\x+\e_2+j\e_1) - \sum_{j = x_1+2}^{T-x_2} v_1(j) \nonumber
\end{align}
Plugging (\ref{eqn:inthm2})  into (\ref{eqn:inthm1}) we obtain:
\begin{align*}
\sum_{j = 0}^{t(\x)} u_2(\x+j\e_1) 
&\ =\ t(\x) \cdot v_2(x_2+1) + u_1(\x+\e_2)  - \sum_{j = x_1+1}^{T-x_2} v_1(j) + \sum_{j = 0}^{t(\x)-1} u_2(\x+\e_2+j\e_1) \\
&\ \leq\ t(\x) \cdot v_2(x_2+1) + \sum_{j = x_1+1}^{k} v_1(j)  - \sum_{j = x_1+1}^{T-x_2} v_1(j) + \sum_{j = 0}^{t(\x)-1} u_2(\x+\e_2+j\e_1) \\
&\ =\ t(\x) \cdot v_2(x_2+1) - \sum_{j = k+1}^{T-x_2} v_1(j) + \sum_{j = 0}^{t(\x+\e_2)} u_2(\x+\e_2+j\e_1) \\
&\ \leq\ t(\x) \cdot v_2(x_2+1) - \sum_{j = k+1}^{T-x_2} v_1(j) + \sum_{i = x_2+2}^{T-k} \Big[(T-x_1-i+1) \cdot v_2(i) - \sum_{j = k+1}^{T-i+1} v_1(j)\Big] \\
&\ =\ \sum_{i = x_2+1}^{T-k} \Big[(T-x_1-i+1) \cdot v_2(i) - \sum_{j = k+1}^{T-i+1} v_1(j)\Big]
\end{align*}
Here the first inequality follows from Lemma \ref{lem:utilityUpperBounds}. The second equality holds as $t(\x+\e_2) = t(\x)-1$.
The second inequality follows from the induction hypothesis because $P^{t(\x)-1}$ is an equilibrium path from $\x+\e_2$ to $(k,T-k)$.
\end{proof}

Note that given an equilibrium path $P$ from $\0$ to $(k,T-k)$, Theorem~\ref{thm:pathValid} and Corollary~\ref{cor:nonNegUtils} imply a class of valid inequalities 
corresponding to each node $x^t$ of $P$. However these inequalities depend on the specific form of $P$ and we want 
inequalities valid for \emph{every} equilibrium path from $\0$ to $(k,T-k)$. 
The following theorem provides such valid inequalities.

\begin{theorem}\label{thm:validInequalities}
Suppose $P$ is an equilibrium path from $\0$ to $(k,T-k)$. 
Then for any $0 \leq \ell < T-k$:
\begin{equation}\label{eqn:validInequalities}
\sum_{i = \ell+1}^{T-k} \Big[(T-i+1) \cdot v_2(i) - \sum_{j = k+1}^{T-i+1} v_1(j)\Big] \geq 0
\end{equation}
\end{theorem}

\begin{proof}
Let $P$ be an equilibrium path from $\0$ to $(k,T-k)$, and $0 \leq \ell < T-k$ an integer. Observe that $P$ must contain a 
decision node $(x_\ell, \ell)$; otherwise, the endpoint of the path $P$ could not have been $(k,T-k)$. Now, by  Theorem~\ref{thm:pathValid} and Corollary~\ref{cor:nonNegUtils}:
$$\sum_{i = \ell +1}^{T-k} \Big[(T-x_\ell-i+1) \cdot v_2(i) - \sum_{j = k+1}^{T-i+1} v_1(j)\Big] \geq 0$$
On the other hand, as incremental valuations are non-negative and $x_\ell \geq 0$:
$$ \sum_{i = \ell + 1}^{T-k} x_\ell \cdot v_2(i) \geq 0$$
Summing up the two inequalities yields the desired result.
\end{proof}

To conclude this section, we write inequality~(\ref{eqn:validInequalities}) in a more amenable form. 
Specifically, for any given $T \in \mathbb{N}$, $0 \leq k \leq T$ and $0 \leq \ell < T-k$, an equivalent formulation is:
\begin{equation}\label{eqn:valid}
\sum_{i = \ell+1}^{T-k} (T-i+1) \cdot v_2(i) - \sum_{i = k + 1}^{T-\ell} (T-i-\ell+1) \cdot v_1(i) \geq 0
\end{equation}
In the remainder of this paper, we will show that the addition of the valid inqualities~(\ref{eqn:valid}) will allow us to obtain sharp
bounds on the efficiency of sequential auctions for both concave and general valuation functions. 

%% file: concave.tex
\section{The Price of Anarchy with Concave Valuation Functions}\label{sec:concave}

In this section, we prove that the price of anarchy is exactly $\big(1-\frac{1}{e}\big)$ when both buyers have
concave valuation functions. We begin by deriving a lower bound conditional on the final allocation.

\begin{theorem}\label{thm:concaveLP-lower}
Let each buyer have a non-decreasing, concave valuation function.
If $(T,0)$ is an optimal allocation and $P$ is an 
equilibrium path from $\0$ to $(k,T-k)$ then:
$$\Gamma(P) \geq \frac{1}{T} \Bigg( k + \sum_{j = 1}^{T-k} \frac{j}{k+j}\Bigg)$$
\end{theorem}
\begin{proof}
Recall buyer~$i$ has a concave valuation function if and only if the incremental valuations satisfy $v_i(j)\ge v_i(j+1)$ 
for all $1\le j\le T-1$.
The addition of these inequalities plus the valid inequalities~(\ref{eqn:valid}) to 
linear program (\ref{LP:unrestricted}) yields:
\begin{alignat}{2}
\text{minimize}\quad \sum_{j = 1}^{k} v_1(j) + \sum_{j = 1}^{T-k} v_2(j) &&&\\
\text{subject to}\ \quad\quad\quad\quad\quad\quad \sum_{j = 1}^T v_1(j) &= 1 &&\nonumber \\
\sum_{j = 1}^l v_1(j) + \sum_{j = 1}^{T-l} v_2(j) &\leq 1 &\quad\quad&\forall\, 0 \leq l < T \nonumber \\
 \sum_{i = \ell+1}^{T-k} (T-i+1) \cdot v_2(i) - \sum_{i = k + 1}^{T-\ell} (T-i-\ell+1) \cdot v_1(i) &\geq 0 &\quad\quad&\forall \ 0 \leq \ell < T-k \nonumber\\
v_i(j+1) - v_i(j) &\leq 0 &\quad\quad&\forall i \in \buyset, j \in [T-1] \nonumber  \\
v_i(j) &\geq 0 							&\quad\quad&\forall i \in \buyset, j \in [T] \nonumber 
\end{alignat}
As discussed, this linear program provides a lower bound on the efficiency when $(T,0)$ is an optimal allocation with $\opt(\0) = 1$
and there is an equilibrium path from $\0$ to $(k,T-k)$.
Further, by weak duality we may lower bound this primal LP by considering its dual LP.

In the dual LP, we assign a dual variable $\sigma_l$ to the welfare constraint for when buyer~$1$ wins $l$ items (for $0 \leq l \leq T$).
We assign a dual variable $\kappa_{i,j}$ for each concavity constraint $j \in [T-1]$ of buyer~$i$; for convenience we also set $\kappa_{i,0} = \kappa_{i,T} = 0$.
Finally, we have a dual variable $\mu_\ell$ for the valid inequalities of type~(\ref{eqn:valid}) for $0 \leq \ell < T-k$.
The dual linear program is then:
\begin{alignat}{2}
\text{maximize}\quad\quad  \sigma_T - \sum_{l = 0}^{T-1} \sigma_l \quad\quad\quad\quad &&& \nonumber\\
\text{subject to}\quad\quad \sum_{l = i}^T \sigma_l - \kappa_{1,i} + \kappa_{1,i-1} &\leq 1 &\quad\quad &\forall \ 1 \leq i \leq k  \label{cons:1}\\
\sum_{l = i}^T \sigma_l - \kappa_{1,i} + \kappa_{1,i-1} - \sum_{\ell = 0}^{T-i} (T-i-\ell+1)\cdot\mu_\ell &\leq  0 &\quad\quad &\forall \ k+1 \leq i \leq T \label{cons:2} \\
\sum_{l = 0}^{T-i} \sigma_l - \kappa_{2,i} + \kappa_{2,i-1} + \sum_{\ell = 0}^{i-1} (T-i+1)\cdot\mu_\ell &\leq 1 &\quad\quad &\forall \ 1 \leq i \leq T-k  \label{cons:3}\\
 \sum_{l = 0}^{T-i} \sigma_l - \kappa_{2,i} + \kappa_{2, i-1} &\leq 0 &\quad\quad &\forall \ T-k+1 \leq i \leq T  \label{cons:4}\\
\sigma_T &\in \mathbb{R} & & \nonumber \\
\sigma_l &\leq 0  &\quad\quad &\forall \ 0 \leq l < T \nonumber \\
\kappa_{i,j} &\leq 0 &\quad\quad&\forall i \in \buyset, j \in [T-1]  \label{cons:5}\\
\kappa_{i,j} &= 0 &\quad\quad&\forall i \in \buyset, j \in \{0,T\}\nonumber \\
\mu_{\ell} &\geq 0 &\quad\quad&\forall \ 0 \leq \ell < T-k \nonumber 
\end{alignat}
Consider the dual solution given by:
\begin{align}\label{dual-solution}
\sigma_T & = \frac{1}{T} \Bigg( k + \sum_{j = 1}^{T-k} \frac{j}{k+j} \Bigg)\nonumber \\
\mu_0 & = \frac{1}{T} \nonumber \\
\mu_\ell & = \frac{1}{T-\ell} - \frac{1}{T-\ell+1} \qquad\qquad\qquad\qquad\quad \forall \ 0 < \ell < T-k \nonumber\\
\kappa_{1,i} & = 
\begin{cases}
 i \cdot (\sigma_T - 1) &\qquad 0 < i \leq k \\
-(T-i) \cdot \sigma_T  + \sum_{j = 0}^{T-i-1} \frac{T-i-j}{T-j} &\qquad k \leq i \leq T-1 \\
\end{cases} 
\end{align}
and all other dual variables set to $0$. 
This solution is dual feasible, with the first four inequalities all holding with equality. To see this, note that for $1 \leq i \leq k$:
\begin{align*}
\sum_{l = i}^T \sigma_l - \kappa_{1,i} + \kappa_{1,i-1} 
&\ =\ \sigma_T + i \cdot (\sigma_T-1) - (i-1) \cdot (\sigma_T-1) 
\ =\  1
\end{align*}
Thus the constraints (\ref{cons:1}) are satisfied with equality. 
To show that the constraints (\ref{cons:2}) hold, observe that for $k+1 \leq i \leq T$:
\begin{align}
\sum_{l = i}^T \sigma_l& - \kappa_{1,i} + \kappa_{1,i-1} - \sum_{\ell = 0}^{T-i} (T-i-\ell+1)\cdot\mu_\ell  \nonumber \\
&\ =\ \sigma_T + (T-i) \cdot \sigma_T - \sum_{j = 0}^{T-i-1} \frac{T-i-j}{T-j} - (T-i+1) \cdot \sigma_T  \nonumber \\
&\ \ \ \ + \sum_{j = 0}^{T-i} \frac{T-i+1-j}{T-j} - \sum_{\ell = 0}^{T-i} (T-i-\ell+1)\cdot\mu_\ell \nonumber \\
&\ =\ \sum_{j = 0}^{T-i} \frac{1}{T-j}   - \sum_{\ell = 0}^{T-i} (T-i-\ell+1)\cdot\mu_\ell \label{eq:cons2a}
\end{align}
Now, by definition of $\mu_\ell$, we have for $k+1 \leq i < T$ that:
\begin{align*}
\sum_{\ell = 0}^{T-i} (T-i-\ell+1)\cdot\mu_\ell 
&\ =\ (T-i+1)\cdot \frac{1}{T} +  \sum_{\ell = 1}^{T-i} (T-i-\ell+1)\cdot  \left(\frac{1}{T-\ell} - \frac{1}{T-\ell+1} \right)\\
&\ =\ (T-i+1)\cdot \frac{1}{T} +  \sum_{\ell = 1}^{T-i} (T-\ell+1)\cdot  \left(\frac{1}{T-\ell} - \frac{1}{T-\ell+1} \right) \\
& \ \ \ \ - i\cdot  \sum_{\ell = 1}^{T-i} \left(\frac{1}{T-\ell} - \frac{1}{T-\ell+1} \right)\\
&\ =\ (T-i+1)\cdot \frac{1}{T} +\sum_{\ell=1}^{T-i} \frac{1}{T-\ell} - i\cdot   \left(\frac{1}{i} - \frac{1}{T} \right)\\
&\ =\  \frac{1}{T} +\sum_{\ell=1}^{T-i} \frac{1}{T-\ell} \\
&\ =\  \sum_{\ell=0}^{T-i} \frac{1}{T-\ell} 
\end{align*}
Note also for $i=T$ it also holds that:
$$\sum_{\ell = 0}^{T-i} (T-i-\ell+1)\cdot\mu_\ell \ =\  \mu_0 \ =\  \frac{1}{T} \ =\ \sum_{\ell=0}^{T-i} \frac{1}{T-\ell}$$
Thus, plugging this into (\ref{eq:cons2a}), we have for $k+1 \leq i \leq T$:
\begin{align*}
& \ \ \  \sum_{l = i}^T \sigma_l - \kappa_{1,i} + \kappa_{1,i-1} - \sum_{\ell = 0}^{T-i} (T-i-\ell+1)\cdot\mu_\ell \ =\ 0
\end{align*}
So the constraints (\ref{cons:2}) are satisfied with equality. Next we show the constraints (\ref{cons:3}) also hold with equality. 
Because, the sum of $\mu_\ell$'s are telescoping, we obtain for $1 \leq i \leq T-k$:
\begin{align*}
\sum_{l = 0}^{T-i} \sigma_l & - \kappa_{2,i} + \kappa_{2,i-1} + \sum_{\ell = 0}^{i-1} (T-i+1)\cdot\mu_\ell \\ 
& = (T-i+1) \cdot \Bigg( \frac{1}{T} +\Big(\frac{1}{T-1} - \frac{1}{T}\Big) +  \Big(\frac{1}{T-2} - \frac{1}{T-1}\Big) + ... + \Big(\frac{1}{T-i+1} - \frac{1}{T-i+2}\Big) \Bigg) \\
& = (T-i+1) \cdot \frac{1}{T-i+1} \\
& = 1
\end{align*}
Furthermore the constraints (\ref{cons:4}) hold with equality. That is, $\sum_{l = 0}^{T-i} \sigma_l - \kappa_{2,i} + \kappa_{2, i-1} = 0$  for
$T-k+1 \leq i \leq T$, since every term of the equation is equal to zero. 

Now, note that we have provided two alternative definitions for $\kappa_{1,k}$ in the dual solution. They are indeed equivalent, as:
\begin{align}\label{k-okay}
k \cdot (\sigma_T-1) - \left( -\left(T-k\right) \cdot \sigma_T + \sum_{j = 0}^{T-k-1} \frac{T-k-j}{T-j} \right) 
& = T \cdot \sigma_T - k + \sum_{j = 1}^{T-k} \frac{j}{k+j} \nonumber\\
& = \left( k + \sum_{j = 1}^{T-k} \frac{j}{k+j} \right) - k + \sum_{j = 1}^{T-k} \frac{j}{k+j} \nonumber \\
& = 0
\end{align}
Finally, it remains to show that the signs of the dual variables are correct. This is trivial to see, except for the variables $\kappa_{1,i}$.
To verify these, observe first that $\sigma_T$ is the average of $T$ numbers in $[0,1]$:
$$\sigma_T = \frac{1}{T} \left( k \cdot 1 + \sum_{j = 1}^{T-k} \frac{j}{k+j} \right) \leq 1$$
Hence $\kappa_{1,i} = i \cdot (\sigma_T-1)\leq 0$ for $0 < i \leq k$. 
 It remains to show that $\kappa_{1,i} \leq 0$ for $k < i < T-1$. To see this, note that:
\begin{align*}
\sigma_T 
&\ =\ \frac{1}{T} \cdot\left( k \cdot 1 + \sum_{j = 1}^{T-k} \frac{j}{k+j} \right) \\
&\ \geq\ \frac{1}{T-k} \cdot \sum_{j = 1}^{T-k} \frac{j}{k+j} \\
&\ \geq\ \frac{1}{T-i} \cdot\sum_{j = 1}^{T-i} \frac{j}{k+j} \\
&\ \geq\ \frac{1}{T-i} \cdot\sum_{j = 1}^{T-i} \frac{j}{i+j} \\
&\ =\ \frac{1}{T-i} \cdot \sum_{j = 0}^{T-i-1} \frac{T-i-j}{T-j}
\end{align*}
Here, we interpret $\sigma_T$ as an average of a set of numbers, and lower bound $\sigma_T$ by repeatedly 
considering sets of numbers with lower mean. Noting that $\frac{j}{k+j} < 1$ for any $j$, we obtain the first inequality. 
The second inequality, in turn, arises from the facts that $\frac{j}{k+j}$ is increasing in $j$ and $T-k \geq T-i$. 
The third inequality holds since $i>k$. The last equality is then obtained by rearranging the sum. Then we have, for $k < i < T-1$:
$$\kappa_{1,i} \ =\  -(T-i) \cdot \sigma_T  + \sum_{j = 0}^{T-i-1} \frac{T-i-j}{T-j} \ \leq\  0$$

Thus the constraints (\ref{cons:5}) hold and the dual solution is feasible, as claimed. 
This dual solution has value $\sigma_T=\frac{1}{T}\big( k + \sum_{j = 1}^{T-k} \frac{j}{k+j}\big)$.
It follows that the efficiency is at least $\frac{1}{T}\big( k + \sum_{j = 1}^{T-k} \frac{j}{k+j}\big)$.
\end{proof}

\begin{theorem}\label{thm:concaveLP-upper}
There exists a $2$-buyer sequential auction with the following properties: both buyers have non-decreasing, concave valuation functions, 
the allocation $(T,0)$ maximizes social welfare and there is an equilibrium path $P$ from $\0$ to $(k,T-k)$ with:
$$\Gamma(P) = \frac{1}{T} \Bigg( k + \sum_{j = 1}^{T-k} \frac{j}{k+j}\Bigg)$$
\end{theorem}
\begin{proof}
Consider a sequential auction with the following valuation profiles:
\begin{align*}
v_1(j) & = 1 \qquad\qquad\qquad 1 \leq j \leq T \\
v_2(j) & = \begin{cases}
\frac{T-k-j+1}{T-j+1} &\quad 1 \leq j \leq T-k \\
0 &\quad \text{else}
\end{cases}
\end{align*}
Observe that the unique optimal allocation is $(T,0)$ with social welfare $\opt(\0) = T$. 
Meanwhile, $\sw(k|\0) = k + \sum_{j = 1}^{T-k} \frac{j}{k+j}$. Therefore, it suffices to show that there exists an 
equilibrium path from $\0$ to $(k,T-k)$. Computation of the forward utilities yields, for any decision node $\x=(x_1,x_2)$:
\begin{align*}
u_1(x_1,x_2) & = (T-x_1-x_2) \cdot \big(1-v_2(x_2+1)\big) \\
u_2(x_1,x_2) & = 0
\end{align*}
In particular, the bidding strategies are:
\begin{align*}
b_1(\x) & = \begin{cases}
1-\frac{(T-x_1-x_2-1) \cdot k}{(T-x_2-1) \cdot (T-x_2)} &\quad x_2 < T-k \\
1 &\quad x_2 \geq T-k \end{cases} \\
b_2(\x) & = v_2(x_2+1) \\
& = \begin{cases}
1 - \frac{k}{T-x_2} &\quad x_2 < T-k \\
0 &\quad x_2 \geq T-k
\end{cases}
\end{align*}
This implies that, for any $0 \leq \ell < T-k$, $b_1(0,\ell) = b_2(0,\ell)$. 
At any other decision node $\x$, we have $b_1(\x) > b_2(\x)$. Hence by breaking all ties 
in favour of buyer~$2$, we obtain an equilibrium path $P$ from $\0$ to $(k,T-k)$ on which buyer~$2$ 
wins the first $T-k$ items and buyer~$1$ wins the last $k$ items.
\end{proof}

These conditional bounds readily extend to an asymptotically tight constant lower bound for efficiency.
\begin{theorem}
Given non-decreasing, concave valuation functions. For any $T \in \mathbb{N}$, any equilibrium 
path $P$ from $\0$ has efficiency at least $1-\frac{1}{e}$. This bound is asymptotically tight 
as $T \rightarrow \infty$.
\end{theorem}

\begin{proof}
Fix $T$ and let $P$ be an equilibrium path from $\0$. By Corollary~\ref{cor:pathEfficiencyBound}, we may assume
that $(T,0)$ is the unique optimal allocation. If buyer~$1$ wins $T$ 
items then $\Gamma(P) = 1 > 1 - \frac{1}{e}$. So suppose that buyer~$1$ wins $k < T$ items. By 
Theorem~\ref{thm:concaveLP-lower}, we have $\Gamma(P) \geq \frac{1}{T}\big( k + \sum_{j = 1}^{T-k} \frac{j}{k+j}\big)$.

Next, observe that:
\begin{align*}
\frac{1}{T}\cdot \sum_{j = 1}^{T-k} \frac{j}{k+j} 
&\ =\ \sum_{j = 0}^{T-k-1} \frac{1}{T} \cdot \frac{T-k-j}{T-j} \\
&\ =\ \sum_{j = 0}^{T-k-1} \frac{1}{T} \cdot \frac{1-k/T-j/T}{1-j/T} \\
&\ \geq\ \int_{0}^{1-k/T} \frac{1-k/T-x}{1-x} \,\mathop{dx}
\end{align*}
The inequality holds as the second line is an {\em upper Darboux sum}. Therefore:
\begin{align*}
\Gamma(P) 
&\ \geq\ \min_{k \in [T] \cup \{0\}} \,\frac{1}{T}\cdot \Bigg( k + \sum_{j = 1}^{T-k} \frac{j}{k+j}\Bigg) \\
&\ \geq\ \min_{k \in [T] \cup \{0\}} \,\frac{k}{T} + \int_{0}^{1-k/T} \frac{1-k/T-x}{1-x} \,\mathop{dx}\\
&\ \geq\ \inf_{\alpha \in [0,1]} \,\alpha + \int_{0}^{1-\alpha} \frac{1-\alpha-x}{1-x} \,\mathop{dx} \\
&\ =\ \inf_{\alpha \in [0,1]}\, 1 + \alpha \ln \alpha
\end{align*}
The infimum is attained for $\alpha = \frac{1}{e}$, with value $1-\frac{1}{e}$. Setting $k = \lfloor T/e \rfloor$ for the valuations 
given in Theorem~\ref{thm:concaveLP-upper} shows the asymptotic tightness of the bound.
\end{proof}

%% file: general.tex
\section{The Price of Anarchy with General Valuation Functions}\label{sec:general}

In this section we consider the case of general non-decreasing valuation functions. 
We show that the price of anarchy is then exactly $1/T$.  In particular, this bound is no longer a constant
but deteriorates linearly with the number of items for sale in the auction.
This value was first identified in \cite{BBB08} as an upper bound for 
the price of anarchy for general non-decreasing valuations.
Again, we begin with the lower bound, and then present the matching upper bound.
\begin{theorem}\label{thm:gen-upper}
Let the buyers have non-decreasing valuation functions. Then any equilibrium path has 
efficiency at least $1/T$, where $T$ is the number of items.
\end{theorem}
\begin{proof}
By induction on $T$. The bound holds for the base case $T = 1$ because single-item
second-price auctions have full efficiency. Now consider $T > 1$ and suppose the valuations are 
such that there exists an equilibrium path $P$ from $\0$ with $\Gamma(P) \leq 1/T$. Note it cannot be 
that $\Gamma(P) \geq \Gamma(P^{T-1})$; otherwise, by the induction hypothesis, $\Gamma(P) \geq 1/(T-1)$. 
Therefore, by Corollary~\ref{cor:pathEfficiencyBound}, we may assume that the unique optimal allocation is $(T,0)$
and that buyer~$2$ wins $T-k > 0$ items on the equilibrium path $P$. 

To lower bound $\Gamma(P)$, we add the valid inequalities~(\ref{eqn:valid}) to the linear program (\ref{LP:unrestricted}).
This gives: 
\begin{alignat}{2}
\text{minimize}\quad \sum_{j = 1}^{k} v_1(j) + \sum_{j = 1}^{T-k} v_2(j) &&&\\ 
\text{subject to}\ \quad\quad\quad\quad\quad\quad \sum_{j = 1}^T v_1(j) &= 1 &&\nonumber \\
\sum_{j = 1}^l v_1(j) + \sum_{j = 1}^{T-l} v_2(j) &\leq 1 &\quad\quad&\forall\, 0 \leq l < T \nonumber \\
 \sum_{i = \ell+1}^{T-k} (T-i+1) \cdot v_2(i) - \sum_{i = k + 1}^{T-\ell} (T-i-\ell+1) \cdot v_1(i) &\geq 0 &\quad\quad& \forall \ 0 \leq \ell < T-k \nonumber\\
v_i(j) &\geq 0 							&\quad\quad&\forall i \in \buyset, j \in [T] \nonumber 
\end{alignat}
Again to lower bound this primal LP we consider its dual LP.
We assign a dual variable $\sigma_l$ to the welfare constraint for when buyer~$1$ wins $l$ items (for $0 \leq l \leq T$).
We have a dual variable $\mu_\ell$ for the valid inequalities of type~(\ref{eqn:valid}) for $0 \leq \ell < T-k$.
The dual linear program is then:
\begin{alignat*}{2}
\text{maximize}\quad\quad  \sigma_T - \sum_{l = 0}^{T-1} \sigma_l &&& \ \\
\text{subject to}\quad\quad\quad\quad \ \sum_{l = i}^T \sigma_l &\leq 1 &\quad\quad &\forall \ 1 \leq i \leq k \\
\sum_{l = i}^T \sigma_l - \sum_{\ell = 0}^{T-i} (T-i-\ell+1)\cdot \mu_\ell &\leq  0 &\quad\quad &\forall \ k+1 \leq i \leq T \\
\sum_{l = 0}^{T-i} \sigma_l + \sum_{\ell = 0}^{i-1} (T-i+1)\cdot \mu_\ell &\leq 1 &\quad\quad &\forall \ 1 \leq i \leq T-k \\
 \sum_{l = 0}^{T-i} \sigma_l &\leq 0 &\quad\quad &\forall \ T-k+1 \leq i \leq T \\
\sigma_T &\in \mathbb{R} & &\\
\sigma_l &\leq 0  &\quad\quad &\forall \ 0 \leq l < T \\
\mu_{\ell} &\geq 0 &\quad\quad&\forall \ 0 \leq \ell < T-k 
\end{alignat*}
Now consider setting $\sigma_T = \mu_0 = 1/T$ and all other variables $0$. It is easy to verify that this is dual feasible 
and has objective value $1/T$. This implies that $\Gamma(P) \geq 1/T$ as desired.
\end{proof}

\begin{theorem}\label{thm:gen-lower}
There exists a $2$-buyer sequential auction with the following properties: both buyers have non-decreasing valuation functions, 
the allocation $(T,0)$ maximizes social welfare and there is an equilibrium path $P$ from $\0$ with:
$$\Gamma(P) = \frac{1}{T}$$
\end{theorem}

\begin{proof}
Consider a sequential auction with the following valuations profiles:
\begin{align*}
v_1(j) & = \begin{cases}
0 & j < T \\
1 & j = T
\end{cases} \\
v_2(j) & = \begin{cases}
1/T & j = 1 \\
0 & j > 1
\end{cases}
\end{align*}

With the given valuation profile, the optimal allocation is $(T,0)$ with a welfare of $1$, while any other allocation has social 
welfare $1/T$. Solving for forward utilities by backwards induction yields, for any decision node $\x=(x_1, x_2)$:
\begin{align*}
u_1(x_1,x_2) & = \begin{cases}
\frac{x_1}{T} & x_2 = 0 \\
0 & x_2 > 0
\end{cases} \\
u_2(x_1,x_2) & = 0
\end{align*}
In particular, $b_1(\0) = b_2(\0) = 1/T$, so buyer~$1$ and $2$ tie at decision node $\0$. Then by breaking the tie in favour of 
buyer~$2$, there exists an equilibrium path from $\0$ which awards at least one item to buyer~$2$, attaining an efficiency of $1/T$.
\end{proof}

\begin{theorem}\label{thm:unrestrictedLP}
The price of anarchy for $2$-buyer sequential auctions with non-decreasing valuations is exactly~$1/T$.~\qed
\end{theorem}